\documentclass[journal]{IEEEtran}

\usepackage{amsmath,amssymb,amsfonts, amsthm}
\usepackage{color}
\usepackage{algorithm}
\usepackage[noend]{algorithmic}
\usepackage{graphicx}
\usepackage{verbatim}
\usepackage{listings}	
\usepackage[caption=false,font=footnotesize]{subfig}	
\usepackage{setspace}
\usepackage[]{epsfig,color}
\usepackage[nice]{nicefrac}
\usepackage{latexsym}
\usepackage{fancyhdr}
\usepackage{colordvi,alltt}
\usepackage{multirow}
\usepackage{tabularx,threeparttable}
\usepackage{fancybox}
\usepackage{calc}
\usepackage{url}
\usepackage{times}
\usepackage[numbers]{natbib}
\usepackage{multicol}
\usepackage[bookmarks=true]{hyperref}
\usepackage{xspace}
\usepackage{breqn}

\author{Aviv~Adler, Mark~de~Berg, Dan~Halperin, and~Kiril~Solovey~
  \thanks{A.\ Adler is with the Electrical Engineering and Computer
    Science department at MIT, Massachusetts, USA. The work has been
    carried out in part during his visit to Tel Aviv University,
    enabled by the generous Melvin M.\ Goldberg Fellowship for
    Research in Israel.}  \thanks{M.\ de Berg is with the Department
    of Mathematics and Computing Science, TU Eindhoven, the
    Netherlands.}  \thanks{D.\ Halperin and K.\ Solovey are with the
    Blavatnik School of Computer Science, Tel-Aviv University,
    Israel. Their work has been supported in part by the 7th Framework
    Programme for Research of the European Commission, under FET-Open
    grant number 255827 (CGL---Computational Geometry Learning), by
    the Israel Science Foundation (grant no. 1102/11), by the
    German-Israeli Foundation (grant no. 1150-82.6/2011), and by the
    Hermann Minkowski--Minerva Center for Geometry at Tel Aviv
    University.}}

\newcommand{\ignore}[1]{}

\def\F{\mathcal{F}}

\def\L{\mathcal{L}}
\def\O{\mathcal{O}}

\def\E{\mathcal{E}}

\def\G{\mathcal{G}}

\def\T{\mathcal{T}}
\def\L{\mathcal{L}}

\def\V{\mathcal{V}}

\def\D{\mathcal{D}}
\def\W{\mathcal{W}}

\def\dD{\mathbb{D}}

\def\obs{\mathrm{obs}}

\newcommand{\defeq}{%
  \mathrel{\vbox{\offinterlineskip\ialign{%
    \hfil##\hfil\cr
    $\scriptscriptstyle\triangle$\cr
    $=$\cr
}}}}
\def\Int{\mathrm{Int}}

\def\Reals{\mathbb{R}}
\renewcommand{\leq}{\leqslant}
\renewcommand{\geq}{\geqslant}
\newcommand{\compl}{\mathrm{Compl}}

\newcommand{\Cpp}{C\raise.08ex\hbox{\tt ++}\xspace}

\newcommand{\pspace}{{\sc pspace}\xspace}
\newcommand{\np}{{\sc np}\xspace}

\newcommand{\etal}{{et~al.}\xspace}

\newboolean{ShowTODO}
\setboolean{ShowTODO}{true}

\ifthenelse{\boolean{ShowTODO}}{%

  \def\marrow{{\raggedright\footnotesize $\longleftarrow$}}

  \def\kiril#1{\textcolor{blue}{{\sc Kiril says: }{\marrow\sf #1}}}

  \def\aviv#1{\textcolor{red}{{\sc Aviv says: }{\marrow\sf #1}}}

  \def\markdb#1{\textcolor{red}{{\sc Mark says: }{\marrow\sf #1}}}

  \def\danny#1{\textcolor{red}{{\sc Danny says: }{\marrow\sf #1}}}

}{%

  \def\kiril#1{}

  \def\aviv#1{}

  \def\markdb#1{}

  \def\danny#1{}

}

\newtheorem{defin}{Definition}
  
\newtheorem{theorem}[defin]{Theorem}
  
\newtheorem{lemma}[defin]{Lemma}
  
\newtheorem{propo}[defin]{Proposition}
  
\newtheorem{coro}[defin]{Corollary}
  
\newtheorem{obse}{Observation}
  \newenvironment{observation}{\begin{obse} \sl}{\end{obse}}

\title{Efficient Multi-Robot Motion Planning for Unlabeled Discs in Simple
Polygons}

\begin{document}

\maketitle

\begin{abstract}
  We consider the following motion-planning problem: we are given $m$
  unit discs in a simple polygon with $n$ vertices, each at their own
  start position, and we want to move the discs to a given set of $m$
  target positions.  Contrary to the standard (labeled) version of the
  problem, each disc is allowed to be moved to any target position, as
  long as in the end every target position is occupied. We show that
  this unlabeled version of the problem can be solved in
  $O\left(n\log n+mn+m^2\right)$ time, assuming that the start and
  target positions are at least some minimal distance from each
  other. This is in sharp contrast to the standard (labeled) and more
  general multi-robot motion planning problem for discs moving in a
  simple polygon, which is known to be strongly \np-hard.
\end{abstract}




\section{Introduction} \label{sec:intro}

%
%
The \emph{multi-robot motion-planning problem} is to plan the motions
of several robots operating in a common workspace. In its most basic
form, the goal is to move each robot from its start position to some
designated target position, while avoiding collision with obstacles in
the environment and with other robots.  Besides its obvious relevance
to robotics, the problem has various other applications, for example
in the design of computer games or crowd simulation.  Multi-robot
motion planning is a natural extension of the single-robot motion
planning problem, but it is much more complex due to the high number
of \emph{degrees of freedom} that it entails, even when the individual
robots are as simple as discs.

\subsection{Related work}
One of the first occurrences of the multi-robot motion-planning
problem in the computational-geometry literature can be found in the
series of papers on the \emph{Piano Movers' Problem} by Schwartz and
Sharir.  They first considered the problem in a general
setting~\cite{ss-pm2} and then narrowed it down to the case of disc
robots moving amidst polygonal obstacles~\cite{ss-pm3}.  In the latter
work an algorithm was presented for the case of two and three robots,
with running time of $O(n^3)$ and $O(n^{13})$, respectively, where $n$
is the complexity of the workspace.  Later Yap~\cite{c-cms84} used the
\emph{retraction method} to develop a more efficient algorithm, which
runs in $O(n^2)$ and $O(n^3)$ time for the case of two and three
robots, respectively.  Several years afterwards, Sharir and
Sifrony~\cite{ss-cmp91} presented a general approach based on
\emph{cell decomposition}, which is capable of dealing with various
types of robot pairs and which has a running time of
$O(n^2)$. Moreover, several techniques that reduce the effective
number of degrees of freedom of the problem have been
proposed~\cite{avbsv-mpfmr,bslm-cppmr}.

When the number of robots is no longer a fixed constant, the
multi-robot motion-planning problem becomes hard.
Hopcroft~\etal~\cite{hss-cmpmio} showed that even in the relatively
simple setting of $n$~rectangular robots moving in a rectangular
workspace, the problem is already \pspace-hard. Moreover, Spirakis and
Yap~\cite{sy-snp84} showed that the problem is strongly \np-hard for
disc robots in a simple polygon.

In recent years, multi-robot planning has attracted a great deal of
attention from the robotics community. This can be mainly attributed
to two reasons. First, it is a problem of practical
importance. Second, the emergence of the \emph{sampling-based
  techniques}, which are relatively easy to implement, yet are highly
effective. These techniques attempt to capture the connectivity of the
configuration space through random sampling~\cite{kslo-prm, l-rert}.
Although sampling-based algorithms are usually incomplete---they are
not guaranteed to find a solution---they tend to be very efficient in
practice. Hence, they are considered the method-of-choice for
motion-planning problems that involve many degrees of freedom.  While
sampling-based tools for a single robot can be applied directly to the
multi-robot problem by considering the group of robots as one large
\emph{composite robot}~\cite{sl-upp}, there is a large body of work
that attempts to exploit the unique properties of the multi-robot
problem~\cite{hh-hmp,SHH12,ssh-fne13,so-cppmr,wc-cmpp11,wmc-ppp12}.

The aforementioned results deal exclusively with the classical
formulation of the multi-robot problem, where the robots are distinct
and every robot is assigned a specific target position.  The
\emph{unlabeled} variant of the problem, where all the robots are
assumed to be identical and thus interchangeable, was first considered
by Kloder and Hutchinson~\cite{kh-pppi05}, who devised a
sampling-based algorithm for the problem. Recently a generalization of
the unlabeled problem---the \emph{$k$-color motion-planning
  problem}---has been proposed, in which there are several groups of
interchangeable robots~\cite{sh-kcmr}.  Turpin~\etal~\cite{tmk-cap13}
considered a special setting of the unlabeled problem with disc
robots, namely where the collection of free configurations surrounding
every start or target position is star-shaped. This condition allows
them to devise an efficient algorithm that computes a solution in
which the maximum path length is minimized.  Unfortunately the
star-shapedness condition is quite restrictive, and in general it will
not be satisfied.

Other related work includes papers that study the number of moves
required to move a set of discs between two sets of positions in an
unbounded workspace, when a move consists of sliding a single
disc---see for example the paper by Bereg~\etal~\cite{bdp-sdp08} which
provides upper and lower bounds for the unlabeled case, or the paper
by Dumitrescu and Jiang~\cite{dj-rdp09} who show that deciding whether
a collection of labeled or unlabeled discs can be moved between two
sets of positions within $k$ steps is \np-hard.  Finally, we mention
the problem of \emph{pebble motions on graphs}, which can be
considered as a discrete variant of the multi-robot motion planning
problem. In this problem, pebbles need to be moved from one set of
vertices of a graph to another, while following a certain set of
rules---see for example~\cite{g-smg84, gh-mcpm, k-cpmg, klb-ftpp13,
  prst-mpg, yl-dofc12}.

\subsection{Our contribution}
Surprisingly, the unlabeled version of the multi-robot motion-planning
problem has hardly received any attention in the
computational-geometry literature.  Indeed, we don't know of any
papers that solve the problem in an exact and complete manner, except
in a restricted setting studied by Turpin~\etal~\cite{tmk-cap13} that
we mentioned above. We therefore study the following basic variant of
the problem: given $m$ unit discs in a simple polygon with
$n$~vertices, each at their own start position, and $m$~target
positions, find collision-free motions for the discs such that at the
end of the motions each disc occupies a target position.  We make the
additional assumption that the given start and target positions are
\emph{well-separated}. More precisely, any two of the given start and
target positions should be at distance at least~4 from each
other. Notice that we only assume this extra separation between the
robots in their static initial and goal placements; we do not assume
any extra separation (beyond non-collision) between a robot and the
obstacles, nor do we enforce any extra separation between the robots
during the motion.  Even this basic version of the problem turns out
to have a rich structure and poses several difficulties and
interesting questions.

By carefully examining the various properties of the problem we show
how to transform it into a discrete pebble-motion problem on graphs.
A solution to the pebble problem, which can be generated with rather
straightforward techniques, can then be transformed back into a
solution to our continuous motion-planning problem. We mention that a
similar transformation was used in~\cite{sh-kcmr} in the context of a
sampling-based method. Using this transformation we are able to devise
an efficient algorithm whose running time is
$O\left(n\log n+mn+m^2\right)$, where $m$ is the number of robots and
$n$ is the complexity of the workspace. To be precise, we show that
our algorithm runs in $O(n\log n+m^2)$ time, and the overall
description length of all the paths to be carried out by the robots
has complexity $O(mn+m^2)$.
As already mentioned, this is in sharp contrast to the standard
(labeled) and more general multi-robot motion planning problem for
discs moving in a simple polygon, which is known to be strongly
\np-hard~\cite{sy-snp84}.

\section{Preliminaries}\label{sec:prel} 
We consider the problem of $m$ indistinguishable unit-disc robots
moving in a simple polygonal workspace~$\W\subset \Reals^2$ with $n$
edges.  We define $\O \defeq \Reals^2\setminus \W$ to be the
complement of the workspace, and we call $\O$ the \emph{obstacle
  space}. Since our robots are discs, a placement of a robot is
uniquely specified by the location of its center. Hence, we will
sometimes refer to points $x\in \W$ as \emph{configurations}, and we
will say that a robot is \emph{at configuration~$x$} when its center
is placed at the point~$x\in W$.  For given $x \in \Reals^2$ and
$r \in \Reals_+$, we define $\D_r(x)$ to be the open disc of radius
$r$ centered at $x$.

We consider the unit-disc robots to be open sets. Thus a robot avoids
collision with the obstacle space if and only if its center is at
distance at least~1 from $\O$, that is, when it is at a configuration
located in the \emph{free space}
$\F \defeq \{x \in \Reals^2 : \D_1(x) \cap \O = \emptyset \}$.
We require the robots to avoid collisions with each other, so if a
robot is at configuration~$x$ then no other robot can be at a
configuration $y\in \Int(\D_2(x))$; here $\Int(X)$ denotes the
interior of the set~$X$.  Furthermore, the notation $\partial (X)$
will be used to refer the boundary of $X$.  We call $\D_2(x)$ the
\emph{collision disc} of the configuration~$x$.

Besides the simple polygon $\W$ forming the workspace, we are also
given sets $S = \{s_1, s_2, ..., s_m\}$ and
$T = \{t_1, t_2, ..., t_m\}$, such that $S, T \subset \F$. These are
respectively the sets of \emph{start} and \emph{target} configurations
of our $m$ identical disc robots. We assume that the configurations in
$S$ and $T$ are \emph{well-separated}:
\begin{quotation}
  For any two distinct configurations $x, y \in S \cup T$ we have
  $\|x - y\| \geq 4$.
\end{quotation}
The problem is now to plan a collision-free motion for our $m$
unit-disc robots such that each of them starts at a configuration in
$S$ and ends at a configuration in $T$.  Since the robots are
indistinguishable (or: \emph{unlabeled}), it does not matter which
robot ends up at which target configuration. Formally, we wish to find
paths $\pi_i:[0,1]\rightarrow \F$, for $1\leq i\leq m$, such that
$\pi_i(0)=s_i$ and $\bigcup_{i=1}^m\pi_i(1)=T$.  Additionally, we
require that the robots do not collide with each other: for every
$1\leq i\neq j\leq m$ and every $\xi\in (0,1)$, we require
$\D_1(\pi_i(\xi))\cap \D_1(\pi_j(\xi))=\emptyset$. Note that the
requirement that the robots do not collide with the obstacle
space~$\O$ is implied by the paths $\pi_i$ being inside the free
space~$\F$.

\section{Basic properties of the free space}\label{sec:b_lemmas}
Recall that the free space $\F\subset \W$ is the set of configurations
at which a robot does not collide with the obstacle space. The free
space may consist of multiple connected components. We denote these
components by $F_1,\ldots,F_q$, where $q$ is the total number of
components.  For any $i \in \{1, 2, ..., q\}$, we let
$S_i \defeq S \cap F_i$ and $T_i \defeq T \cap F_i$.  We assume from
now on that $|S_i| = |T_i|$ for all $1\leq i\leq q$---if this is not
the case, then the problem instance obviously has no solution---and we
define $m_i \defeq |S_i| = |T_i|$ to be the number of robots in~$F_i$.

Before we proceed, we need one more piece of notation.  For any
$x \in \W$, we define $\obs(x)$, the \emph{obstacle set} of $x$, as
$\obs(x) \defeq \{y \in \O : \|x - y\| < 1\}$. In other words,
$\obs(x)$ contains the points in the obstacle space overlapping with
$\D_1(x)$.  Note that $\obs(x) = \emptyset$ for $x \in \F$.

In the remainder of this section we prove several crucial properties
of the free space, which will allow us to transform our problem to a
discrete pebble problem. We start with some properties of individual
components $F_i$, and then consider the interaction between robots in
different components.

\subsection{Properties of a single connected component of $\F$}
We start with a simple observation, for which we provide a proof for
completeness.
\begin{lemma}\label{lem:shape of F_i} 
  Each component $F_i$ is simply connected.
\end{lemma}
\begin{proof}
  Suppose for a contradiction that $F_i$ contains a hole.  Then
  $\compl(\F) \defeq \Reals^2\setminus \F$, the complement of the free
  space, has multiple connected components. One of these is $C_{\O}$,
  the unbounded component containing~$\O$. Let $C$ be another
  component of $\compl(\F)$, and let $x\in C$.  Since $x\not\in \F$,
  there is a point $y\in\O$ with $\|x - y\| < 1$. But then
  $\|x' - y\| < 1$ for any point $x'$ on the segment $\overline{xy}$,
  which implies $\overline{xy}\subset \compl(\F)$ and thus contradicts
  that $C$ and $C_{\O}$ are different components.
\end{proof}
Now consider any $x$ in $\mathbb{R}^2$. Recall that $\D_2(x)$ denotes
the collision disc of~$x$, that is, $\D_2(x)$ is the set of all
configurations $y$ for which another robot placed at~$y$ collides with
a robot at~$x$. We now define $D^*(x)$ to be the part of $\D_2(x)$
that is in the same free-space component as~$x$, that is,
$D^*(x) \defeq \D_2(x)\cap F_i$ where $F_i$ is the free-space
component such that~$x\in F_i$.

The following three lemmas constitute the theoretical basis on which
the correctness and efficiency of our algorithm relies.

\begin{lemma} \label{lem: component self interference} For any
  $x\in \F$, the set $D^*(x)$ is connected.
\end{lemma}

\begin{proof}
  Assume for a contradiction that $D^*(x)$ is not connected.  Let
  $F_i$ be the free-space component containing~$x$. Since by
  definition $x \in D^*(x)$, we can find some $y \in D^*(x)$ that is
  in a different connected component of $D^*(x)$ from $x$. Since
  $y \in D^*(x) \subset \D_2(x)$, the distance between $x$ and $y$ is
  at most~$2$. Hence, any point on the line segment $\overline{xy}$ is
  within a distance of $1$ of either $x$ or $y$. Since $x, y \in F_i$,
  we know that $\overline{xy} \subset \W$, otherwise either $x$ or $y$
  would not be in $\F$. We also know that
  $\overline{xy} \not \subset F_i$, since otherwise $x$ and $y$ would
  not be in different connected components of $D^*(x)$. Because
  $x, y \in F_i$, by definition there exists a simple path
  $\pi \subset F_i$ from $x$ to $y$. Since the workspace is a polygon
  with finite description complexity, we may assume that $\pi$ has
  finite complexity as well, which implies that
  $\pi \cap \overline{xy}$ is composed of finitely many isolated
  points and closed segments. See Figure~\ref{fig:lemmas}~(a) for an
  illustration.

  We now define $x', y'$ as the points on
  $\pi \cap \overline{xy} \subset D^*(x)$ such that $x', y'$ are in
  different connected components of $D^*(x)$ and $\|x' - y'\|$ is
  minimized given the first condition. Let $\pi'$ be the subpath of
  $\pi$ joining $x'$ to $y'$.  Notice that
  $\pi \cap \overline{x'y'} = \{x',y'\}$. Indeed, if there exists a
  point~$z \in \pi \cap \Int(\overline{x'y'})$, then $z$ must be in a
  different connected component of $D^*(x)$ than either $x'$ or $y'$,
  and $\|x' - y'\|$ would not be the minimum. Since $\pi$ is a simple
  path, this means that $\lambda\defeq\pi' \cup \overline{x'y'}$ is a
  simple closed curve.  The area enclosed by $\lambda$ (including
  $\lambda$) will be referred to as $A$. We note that
  $\lambda \subset \W$ since $\pi' \subset \F \subset \W$ and
  $\overline{x'y'} \subset \overline{xy} \subset \W$. This immediately
  implies that $A \subset \W$, since $\W$ is a simple polygon.

  Let $A^* \defeq A \backslash \F$. We claim that
  $A^*\subset \Int(\D_2(x))$, which implies that there exists a path
  in $F_i$ from $x'$ to $y'$ that goes along $\partial(A^*)$ and is
  fully contained in $\D_2(x)$. But this contradicts that $x'$ and
  $y'$ are in different components of $D^*(x)$ and, hence, proves the
  lemma. It thus remains to prove the claim that
  $A^*\subset \Int(\D_2(x))$.

  Note that for any point $z \in A^*$ and any $w \in \obs(z)$ we have
  $\overline{zw} \cap \pi' = \emptyset$, since $\pi' \subset \F$.
  Furthermore, for any $v \in \pi'$ we have $\|w - v\| \geq 1$, and as
  $x',y'\in \pi'$ it follows that $\|w-x'\|\geq 1,\|w-y'\|\geq 1$.
  Assume without loss of generality that $\overline{x'y'}$ is vertical
  and that locally $A$ lies to the right of $\overline{x'y'}$, as in
  Figure~\ref{fig:lemmas}~(a).  Let $K$ be the circle of radius $1$
  that passes through $x'$ and $y'$, and whose center lies to the left
  of $\overline{x'y'}$---such a circle always exists since
  $\|x' - y'\| \leq \|x - y\| \leq 2$. (If $\|x'-y'\| = 2$ then the
  center of the circle lies on $\overline{x'y'}$.) Let $\zeta$ be the
  arc of this circle lying to the right of $\overline{x'y'}$; note
  that this is the shorter of the two arcs joining $x'$ and $y'$ if
  they are of different lengths. Then $A^*$ is a contained entirely
  within the area enclosed by $\zeta$ and $\overline{x'y'}$.
  Furthermore, $A^* \subset \Int(A) \cup \Int(\overline{x'y'})$ since
  $\pi' \subset \F$. Therefore, since $\overline{x'y'}$ is a
  subsegment of $\overline{xy}$ and $\zeta$ cannot cross
  $\partial(\D_2(x))$, it follows that
  \[
  A^* \ \ \subset \ \ (\Int(A) \cup \Int(\overline{x'y'})) \cap
  \Int(\D_2(x)) \ \ \subset \ \ \Int(\D_2(x)),
  \]
  which finishes the proof of the lemma.
\end{proof}

\begin{figure}[bh!]
  \centering
  \subfloat[ ]{\includegraphics[width=0.8\columnwidth] {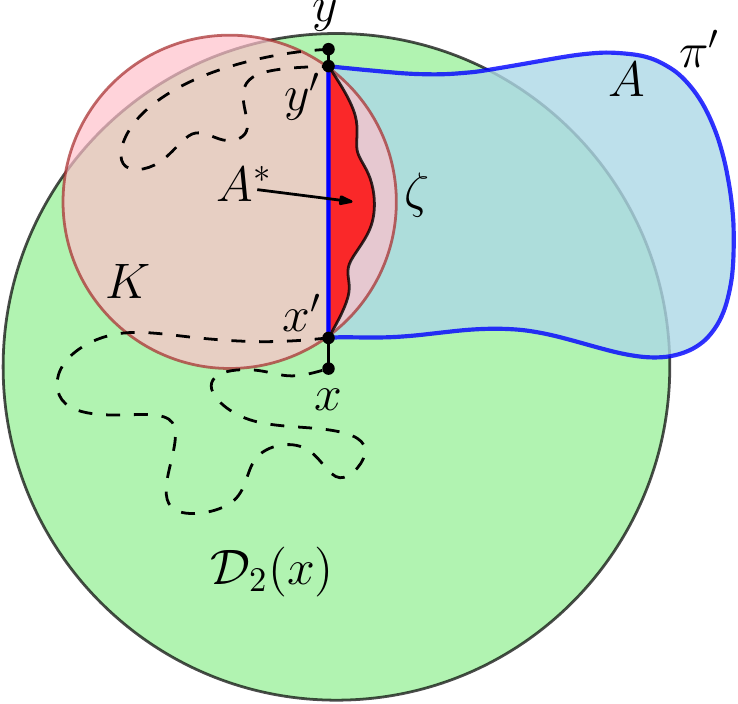}} \\
  \subfloat[ ]{\includegraphics[width=0.8\columnwidth] {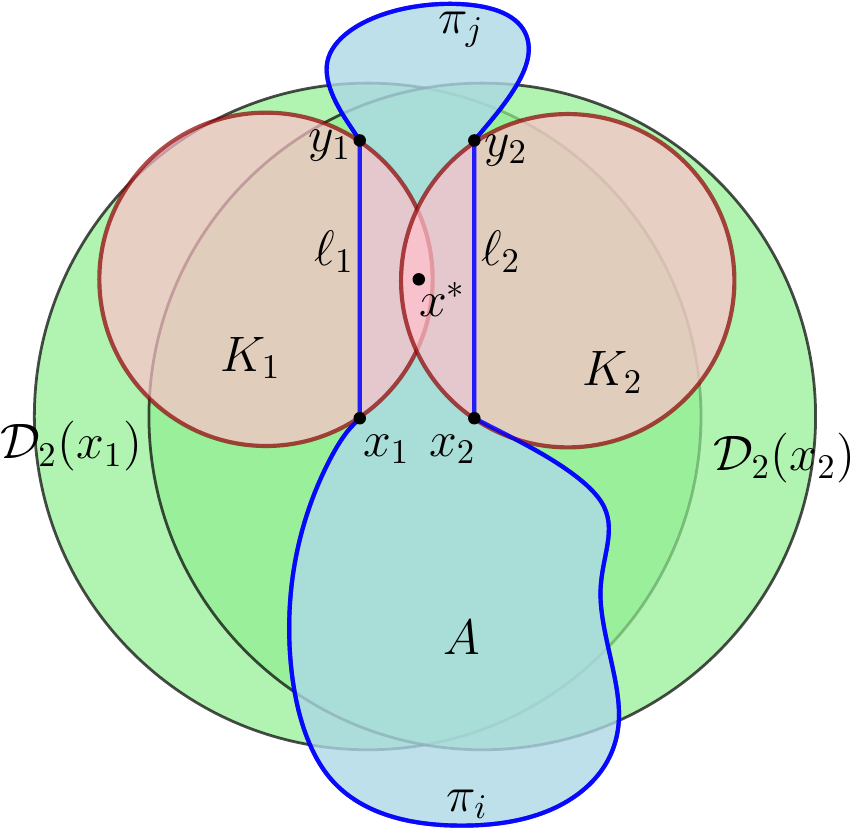}}
  \caption{ (a) An illustration of Lemma~\ref{lem: component self
      interference} . The disc $\D_2(x)$ is drawn in green. The closed
    curve $\lambda$, which consists of the curve $\pi'$ and the
    straight-line $\overline{x'y'}$, is drawn in blue, and $A$
    represents the area that is bounded by $\lambda$. The disc $K$ of
    radius $1$ that touches $x',y'$ is drawn in pink. Note that the
    area $A^*$, which is drawn in red, is contained in $A$. The dashed
    black lines represent $\pi\setminus \pi'$. (b) An illustration of
    Lemma~\ref{lem: no double interference}, and in particular, the
    case where $A^*_1\cap A^*_2\neq \emptyset$. For simplicity of
    presentation, we assume that $\ell_1=\overline{x_1y_1}$ and
    $\ell_2=\overline{x_2y_2}$.}
  \label{fig:lemmas}
  \vspace{-15pt}
\end{figure}

\subsection{Interference between different connected components of $\F$.}
Let $F_i, F_j$ be two distinct components of $\F$, and let $x \in F_i$
be such that $\D_2(x) \cap F_j \neq \emptyset$. We then call $x$ an
\emph{interference configuration from $F_i$ to $F_j$}, and define the
\emph{interference set from $F_i$ to $F_j$} as
$I_{(i,j)} \defeq \{x \in F_i : \D_2(x) \cap F_j \neq \emptyset\}$.
We also define the \emph{mutual interference set} of $F_i, F_j$ as
$I_{\{i,j\}} \defeq I_{(i,j)} \cup I_{(j,i)}$.  Intuitively, an
interference configuration from $F_i$ to $F_j$ is a configuration for
a robot in $F_i$ which could block a path in $F_j$, and the
interference set is the set of all such points. The mutual
interference set of $F_i, F_j$ is the set of all single-robot
configurations in either component which might block a valid
single-robot path in the other component.
\begin{lemma}\label{lem: no double interference} 
  For any mutual interference set $I_{\{i,j\}}$ and any two
  configurations $x_1, x_2 \in I_{\{i,j\}}$ \ we have
  $\D_2(x_1) \cap \D_2(x_2) \neq \emptyset$.
\end{lemma}
\begin{proof}
  The proof is similar in spirit to the proof of Lemma~\ref{lem:
    component self interference} albeit slightly more involved. Assume
  for a contradiction that $x_1, x_2 \in I_{\{i,j\}}$ and
  $\D_2(x_1) \cap \D_2(x_2) = \emptyset$. By definition there exist
  $y_1 \in \D_2(x_1)$ and $y_2 \in \D_2(x_2)$ such that each pair
  $\{x_1, y_1\}, \{x_2, y_2\}$ contains one point in $F_i$ and one
  point in $F_j$. As shown in the proof for Lemma~\ref{lem: component
    self interference}, the segments
  $\overline{x_1 y_1}, \overline{x_2 y_2}$ are entirely contained
  in~$\W$. We may assume that $\overline{x_1 y_1}$ does not cross
  $\overline{x_2 y_2}$, since if it did the crossing point would be in
  $\D_2(x_1) \cap \D_2(x_2)$ and we would be done.  Therefore, there
  exists a simple closed curve $\lambda \subset \W$ composed of the
  union of two simple curves $\pi_i, \pi_j$ and two line segments
  $\ell_1, \ell_2$ such that $\pi_i \subset F_i$ and
  $\pi_j \subset F_j$, and
  $\ell_1 \subset \overline{x_1 y_1},\ell_2 \subset \overline{x_2
    y_2}$.
  Note that both $\ell_1$ and $\ell_2$ have one endpoint in $F_i$ and
  the other in $F_j$; see Figure~\ref{fig:lemmas}~(b) for an
  illustration. The end points of $\ell_1$ consist of $x_1',y'_1$,
  such that $x_1,x'_1$ and $y_1,y'_1$ belong to the same connected
  components, and minimize the distance $\|x'_1-y'_1\|$ ($\ell_2$ is
  similarly defined).

  We refer to the region enclosed by $\lambda$ (including $\lambda$)
  as $A$.  Because $\lambda\subset \W$ and $\W$ is a simple polygon,
  we know that $A \subset \W$. Furthermore, since
  $\pi_i, \pi_j \subset \F$, for any $x \in \Int(A)$ and
  $y \in \obs(x)$ (by definition, $y \in \mathbb{R}^2 \backslash \W$
  so $y \not \in A$; thus,
  $\overline{xy} \cap \lambda \neq \emptyset$), we know that
  $\overline{xy} \cap \pi_i = \overline{xy} \cap \pi_j = \emptyset$.
  Thus, $\overline{xy} \cap \Int(\ell_1) \neq \emptyset$ or
  $\overline{xy} \cap \Int(\ell_2) \neq \emptyset$, or both. Let
  $A^* \defeq A \backslash \F$ and denote by $A^*_1$ the set of
  configurations $x\in A^*$ for which there exists $y \in \obs(x)$
  such that $\overline{xy} \cap \Int(\ell_1) \neq \emptyset$; the
  set~$A^*_2$ is defined in a similar manner, only that now
  $\overline{xy} \cap \Int(\ell_2) \neq \emptyset$. Note that
  $A^*=A^*_1\cup A^*_2$.

  We claim that $A_1^*\cap A_2^*\neq \emptyset$. Indeed, if
  $A_1^*\cap A_2^*=\emptyset$ then there is a path from $x_1$ to $y_1$
  along $\partial(A_1^*)$ that stays in $A\setminus A^*$ and, hence,
  stays in $\F$, which would contradict that $x_1\in F_i$ and
  $ y_1\in F_j$ for $i\neq j$.
  Thus, there exists a point~$x^*\in A_1^*\cap A_2^*$.  We define the
  unit circles $K_1,K_2$ whose boundaries lie on the endpoints of
  $\ell_1,\ell_2$ respectively, and whose centers are located outside
  $A$.  Thus, we have $A^*_1 \subset K_1$ and $A^*_2 \subset K_2$.
  Hence, $x^* \in K_1 \cap K_2$, which implies
  $x^* \in \D_2(x_1) \cap \D_2(x_2)$, so
  $\D_2(x_1) \cap \D_2(x_2) \neq \emptyset$, contradicting our initial
  assumption.
\end{proof}

The next lemma is a generalization of the previous one. Intuitively,
instead of considering a cycle of length~2 among interacting
free-space components, we now consider larger cycles.

\begin{lemma} \label{lem: no circular interference} Let
  $\{\phi(1), \phi(2), ..., \phi(h)\} \subset \{1, 2, ..., q\}$, and
  let $x_1, x_2, ..., x_h$ be points such that for all $i$,
  $x_i \in I_{\{\phi(i), \phi(i+1)\}}$, where
  $\phi(h + 1) \equiv \phi(1)$. (Thus the list is circular with
  respect to its index).  Then there exists some $i \neq j$ such that
  $\D_2(x_i) \cap \D_2(x_j) \neq \emptyset$.
\end{lemma}

\begin{proof}
  This can be proved in a manner completely analogous to the proof of
  Lemma~\ref{lem: no double interference}; we will outline the proof
  here. We assume for a contradiction that
  $\D_2(x_i) \cap \D_2(x_j) = \emptyset$ for all $i \neq j$.  We can
  argue that we can construct a simple closed curve
  $\lambda \subset \W$ passing through
  $F_{\phi(1)}, F_{\phi(2)}, ..., F_{\phi(h)}$ (in that order), which
  is composed of simple closed curves $\pi_i \subset F_{\phi(i)}$ and
  line segments $\ell_i \subset \W$ with endpoints in $F_{\phi(i)}$
  and $F_{\phi(i+1)}$.  We then consider the area $A$ enclosed by
  $\lambda$ and note that $A \subset \W$.  Let
  $A^* \defeq A \setminus \F$. If there exists some simple curve
  $\pi^* \subset A^*$ connecting $\ell_i$ to $\ell_j$ for some
  $i \neq j$, we can show that there exists some $k$ such that
  $\D_2(x_i) \cap \D_2(x_k) \neq \emptyset$, contradicting our
  assumption. Therefore no such $\pi^*$ exists for any $i \neq j$. But
  this means that there exists some simple path
  $\pi' \subset A \cap \F$ which joins $\pi_i$ and $\pi_j$ for some
  $i \neq j$, which contradicts the fact that $\pi_i$ and $\pi_j$
  belong to different components of $\F$.
\end{proof}

\section{Algorithm for a single component}\label{sec:single_cc}
%
%
In this section we consider a single component $F_i$ of $\F$.  We
present an algorithm that solves the problem within $F_i$, ignoring
the possibility that robots in $F_i$ might collide with robots in
other components~$F_j$.  In the next section we will show how to avoid
such collisions without changing the motion plans within the
individual components.  As before we set $S_i \defeq S\cap F_i$ and
$T_i \defeq T\cap F_i$, and assume $|S_i| = |T_i|$.
%
\newcommand{\adg}{motion graph}
\subsection{The \adg}\label{subsec:def_F_i}
The \emph{\adg} $G_i$ of $F_i$ is a graph whose vertices represent
start or target configurations, and whose edges represent
``adjacencies'' between these configurations, as defined more
precisely below.

Recall that for any $x\in F_i$ we defined
$D^*(x) \defeq \D_2(x) \cap F_i$ as the part of the collision disc
of~$x$ inside~$F_i$, and recall from Lemma~\ref{lem: component self
  interference} that $D^*(x)$ is connected.  Moreover, for any two
distinct configurations $x_1, x_2 \in S_i \cup T_i$ we have
$D^*(x_1) \cap D^*(x_2) = \emptyset$, because
$\D_2(x_1) \cap \D_2(x_2) = \emptyset$ by the assumption that the
start and target positions are well-separated.  The vertices of our
\adg~$G_i$ correspond to the start and target configurations in
$S_i\cup T_i$. From now on, and with a slight abuse of notation we
will not distinguish between configurations in $S_i\cup T_i$ and their
corresponding vertices in~$G_i$.

Now consider
$F^*_i\defeq F_i \setminus \bigcup_{x \in S_i \cup T_i} D^*(x)$, the
complement of the collision discs of the given start and target
configurations in~$F_i$. This complement consists of several connected
components, which we denote by $F_i^1,F_i^2,\ldots$.  If the \adg~
$G_i$ contains an edge $(x_1,x_2)$ then there is a component
$F_i^{\ell}$ that is adjacent to both $D^*(x_1)$ and $D^*(x_2)$.  In
other words, two configurations $x_1$ and $x_2$ are connected in $G_i$
if there is a path from $x_1$ to $x_2$ that stays inside $F_i$ and
does not cross the collision disc of any other configuration
$x_3\in S_i\cup T_i$. Figure~\ref{fig:partition_of_F} illustrates the
definition of~$G_i$.  The following observation summarizes the main
property of the \adg.
\begin{observation}\label{le:adg}
  Suppose all robots in $F_i$ are located at a start or target
  configuration in $S_i\cup T_i$, and let $(x_1,x_2)$ be any edge
  in~$G_i$. Then a robot located at $x_1$ can move to $x_2$ without
  colliding with any of the other robots.
\end{observation}
\noindent \emph{Remark.} We could also work with the dual graph of the
partitioning of $F_i$ into cells induced by the collision discs. This
dual graph would, in addition to vertices representing start and
target configurations, also have vertices for the regionss
$F_i^{\ell}$. For the pebble-motion problem discussed below it is
easier to work with the graph as we defined it.  This graph may have
many more edges, but in the implementation of our algorithm described
in Section~\ref{sec:implementation} we avoid computing it explicitly.
\begin{figure}[b]
  \centering \subfloat[ ]{\includegraphics[width=0.5\columnwidth]
    {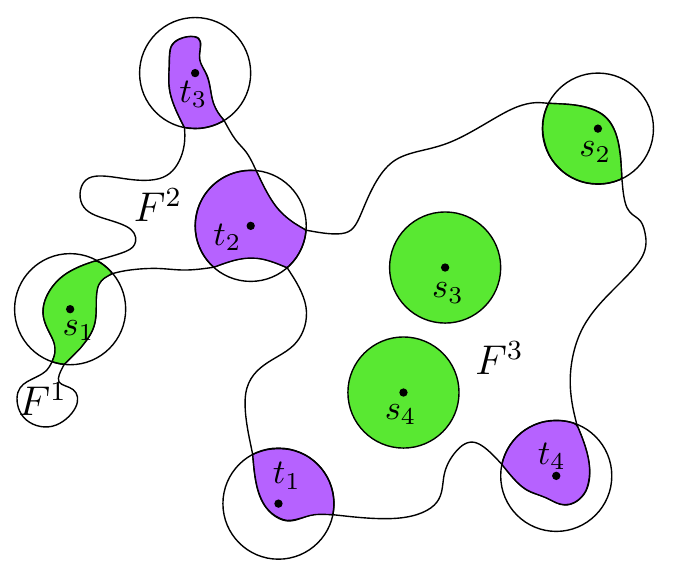}} \subfloat[
  ]{\raisebox{5mm}{\includegraphics[width=0.4\columnwidth]
      {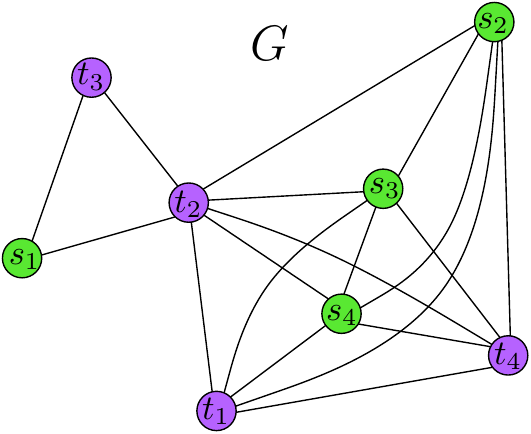}}}
  \caption{(a) A partition of a maximal connected component $F$. The
    start and target positions consist of the elements
    $S'=\{s_1,s_2,s_3,s_4\}, T'=\{t_1,t_2,t_3,t_4\}$, respectively,
    where the areas $D^*(s)$ for $s\in S'$ are drawn in green and
    $D^*(t)$ for $t\in T'$ are drawn in purple. $F^*$ consists of the
    parts $F^1,F^2,F^3$. (b) A \adg~of $F$.}
  \label{fig:partition_of_F}
\end{figure}
%
\subsection{The unlabeled pebble-motion problem}\label{subsec:pebble_motion}
Using the \adg~ $G_i$ we can view the motion-planning problem within
$F_i$ as a pebble-motion problem. (A similar approach was taken
in~\cite{sh-kcmr}, where a sampling-based algorithm for multi-robot
motion planning produces multiple pebble problems by random sampling
of the configuration space.)  To this end we represent a robot located
at configuration $x\in S_i\cup T_i$ by a \emph{pebble} on the
corresponding vertex in $G_i$.  The pebbles are indistinguishable,
like the robots, and they can move along the edges of the graph.  At
the start of the pebble-motion problem for a graph with vertex set
$S_i\cup T_i$, with $|S_i| = |T_i|$, there is a pebble on every vertex
$x\in S_i$. The goal is to move the pebbles such that each pebble ends
up in vertex in $T_i$, under the following conditions: (1) no two
pebbles may occupy the same vertex at the same time, and (2) pebbles
can only halt at vertices, and (3) at most one pebble may move (that
is, be in transit along an edge) at any given time.  We call this
problem the \emph{unlabeled pebble-motion problem}.  The following
lemma follows immediately from Observation~\ref{le:adg}.
\begin{lemma}\label{lem:pebbles to robots} 
  Any solution to the unlabeled pebble-motion problem on $G_i$ can be
  translated into a valid collision-free motion plan for the robots in
  $F_i$.
\end{lemma}
Kornhauser~\cite[Section 3, first lemma]{k-cpmg} proved that the
unlabeled pebble-motion problem is, in fact, always solvable, and he
gave an algorithm to find a solution. Since he did not analyze the
running time of his algorithm, we sketch the solution in the proof of
the lemma below.
\begin{lemma}{\rm\bf\cite{k-cpmg}}\label{lem:pebble} 
  For any graph $G$ with vertex set $S\cup T$ where $|S|=|T|$, there
  exists a solution to the unlabeled pebble-motion problem.  Moreover,
  a solution can be found in $O(|S|^2)$ time.
\end{lemma}
\begin{proof}
  Let $\T_{G}$ be a spanning tree of $G$. The algorithm performs
  $O(|S|)$ \emph{phases}.  In each phase, one or more pebbles may be
  moved and one leaf will be removed from $\T_{G}$, possibly with a
  pebble on it. After the phase ends, the algorithm continues with the
  next phase on the modified tree~$\T_{G}$, until all pebbles have
  been removed and the problem has been solved. A phase proceeds as
  follows.

  If there are leaves $v$ that are target vertices then we select such
  a leaf~$v$.  If $v$ does not yet contain a pebble, we find a pebble
  closest to $v$ in $\T_{G}$---this can be done by a simple
  breadth-first search---and move it to $v$ along the shortest path
  in~$G$. Note that the vertices on the shortest path cannot contain
  other pebbles, since we took a closest pebble. We now remove the
  leaf $v$, together with the pebble occupying it, and end the phase.
  If all leaves in $\T_G$ are start vertices, then let $v$ be such a
  leaf.  If $v$ is not occupied by a pebble it can be removed from
  $\T_{G}$, and the phase ends.  Otherwise a pebble resides in $v$,
  which we move away, as follows. We find the closest unoccupied
  vertex $w$ to $v$ of $\T_{G}$ and move all pebbles on this shortest
  path (including the pebble on $v$) one step closer to $w$, in order
  of decreasing distance from $w$.  After we evacuated~$v$ we remove
  it from $\T_{G}$ to end the phase.

  The algorithm produces paths of total length $O(|S|^2)$, and it can
  easily be implemented to run in $O(|S|^2)$ time.  In some cases
  $\Omega(|S|^2)$ moves are required, e.g., when $G$ is a single path
  with all start positions in the first half of the path and all
  target positions in the second half.
\end{proof}
\begin{lemma}\label{lem: one component}
  Suppose we have an instance of our multi-robot path planning problem
  where $|S_i|=|T_i|$ for every component $F_i$ of the free space
  $\F$.  Then for each $F_i$ there exists a motion plan $\Pi_i$ that
  brings the robots in $F_i$ from $S_i$ to $T_i$, such that they do
  not collide with the obstacle space nor with the other robots
  in~$F_i$.
\end{lemma}
\section{Combining single-component plans}\label{sec:multi_cc}
We now consider possible interactions between robots contained in
different components $F_i$ and $F_j$ of $\F$.  As before, we assume
that $|S_i| = |T_i|$ for all $i$.  We will show that there exists a
permutation $\sigma : \{1, 2, ..., \ell\} \to \{1, 2, ..., \ell\}$
such that we can independently execute the single-component motion
plans for each component $F_i$ as long as we do so in the order
$F_{\sigma(1)}, F_{\sigma(2)}, ..., F_{\sigma(\ell)}$.

To obtain this order, we define a directed graph representing the
structure of $\F$, which we call the \emph{directed-interference
  forest} $\G = (\V, \E)$, where the nodes in $\V$ correspond to the
components $F_i$. We add the directed edge $(F_i, F_j)$ to $\E$ if
either there exists a start position $s \in S$ such that
$s \in I_{(i,j)}$, or there exists a target position $t \in T$ such
that $t \in I_{(j,i)}$. For any $i \in \{1, 2,..., \ell\}$, we
additionally define $N^+(i)$ to be the set of indices of the vertices
in the out-neighborhood of $v_i$; similarly, $N^-(i)$ is defined as
the set of indices of the vertices in the in-neighborhood of $v_i$.

Note that by Lemma~\ref{lem: no double interference} and since $S,T$
are well separated, we cannot have more than one start or target
position in $I_{\{i,j\}}$.  This implies that $\E$ cannot contain both
$(v_i, v_j)$ and $(v_j, v_i)$.  Lemma~\ref{lem: no circular
  interference} and the well-separatedness condition additionally
imply that $\G$ cannot have an undirected cycle. Thus, $\G$ is a
directed forest.

We now produce the desired ordering using $\G$. Consider $F_i \in \V$,
and suppose that for all $j \in N^+(i)$, every robot in $F_j$ is at a
start position, and for all $j \in N^-(i)$, every robot in $F_j$ is at
a target position.  Additionally, suppose that for all
$j \not \in N^+(i) \cup N^-(i)$, every robot in $F_j$ is at a start or
target position.  Then, by the definition of $\G$, no robot is at a
configuration in $I_{\{i,j\}}$ for any $j \neq i$; thus any motion
plan for the robots in $F_i$, such as the one described in
Section~\ref{sec:single_cc}, can be carried out without being blocked
by the robots not in $F_i$.  Hence, if we have an ordering
$\sigma: \{1, 2, ..., \ell\} \to \{1, 2, ..., \ell\}$ such that for
all (directed) edges $(v_i, v_j) \in \E$,
$\sigma^{-1}(i) < \sigma^{-1}(j)$, where $\sigma^{-1}$ is the inverse
permutation of $\sigma$, then we can execute the motion plans for the
robots in $F_{\sigma(1)}, F_{\sigma(2)}, ..., F_{\sigma(\ell)}$ in
that order. Since $\G$ is a directed forest such an ordering can be
produced using topological sorting on the vertices of $\G$. Thus,
combining this result with Lemma~\ref{lem: one component} we obtain:

\begin{theorem}\label{thm: many components}
  Let there be a collection of $m$ unlabeled unit-disc robots in a
  simple polygonal workspace $\W \subset \mathbb{R}^2$, with start and
  target configurations $S, T$ that are well-separated. Then if for
  every maximal connected component $F_i$ of $\F$ (where $\F$ is the
  free space for a single unit-disc robot in $\W$)
  $|S \cap F_i| = |T \cap F_i|$, there exists a collision-free motion
  plan for these robots starting at $S$ which terminates with every
  position of $T$ occupied by a robot.
\end{theorem}

\section{Algorithmic details}\label{sec:implementation}
In this section we fill in a few missing details in the description of
our algorithm. Specifically, we present an efficient method for
generating \adg s and describe a technique for generating
configuration-space paths that correspond to edges in the \adg s. We
also consider the complexity of the various subsets of $\F$ used
throughout the algorithm.

\subsection{Partitioning $\F$}
We analyze the combinatorial complexity of
$\F^*\defeq \F \setminus \bigcup_{x \in S \cup T} D^*(x)$ and
$\dD\defeq\bigcup_{x\in S\cup T}D^*(x)$. 
\begin{lemma}\label{lem:complexity_F*}
  The combinatorial complexity of $\F^*$ is $O(m+n)$.
\end{lemma}
\begin{proof}
  We decompose the complement of the workspace polygon into $O(n)$
  trapezoids---this is doable by standard vertical decomposition.  We
  define a set $X$, which consists of the trapezoids, and in addition
  a collection of $O(m)$ unit discs that are centered at the start and
  target positions. We now observe that the regions in $X$ are
  pairwise interior disjoint (and convex). Hence, it is
  known~\cite{klps-ujr86} that the complexity of the union of the
  regions in $X$, each Minkowski-summed with a unit disc, is linear in
  the number of regions plus the sum of the complexities of the
  original regions. As the result of the Minkowski sum operation of
  $X$ with a unit disc is the the area $\F^*$, we conclude that that
  the complexity of $\F^*$ is $O(m+n)$.
\end{proof}

Note that this upper bound still holds if we consider instead of
$\F^*$ the union of
$F_i^*\defeq F_i \setminus \bigcup_{x \in S_i \cup T_i} D^*(x)$, for
all $1\leq i \leq q$.

\begin{lemma}\label{lem:complexity_D^*}
  The combinatorial complexity of
  $\dD\defeq\bigcup_{x\in S\cup T}D^*(x)$, is $O(m+n)$.
\end{lemma}
\begin{proof}
  Denote by $d\defeq\{d_1,d_2,\ldots\}$ the segments that define
  $\partial(\dD)$. Additionally, denote by $f\defeq\{f_1,f_2,\ldots\}$
  and $f^*\defeq\{f^*_1,f^*_2,\ldots\}$ the segments that define
  $\partial(\F),\partial(\F^*)$, respectively. Note that
  $\partial(\dD)$ consists of segments that are elements of $f,f^*$
  and in addition segments that are subsegments of the elements of
  $f$, denoted by $f'\defeq\{f'_1,f'_2,\ldots\}$. Obviously the
  complexity of the segments of $d$, that are elements of $f$ or
  $f^*$, is bounded by $O(m+n)$. It might happen that the segments of
  $f$ will be split into many subsegments in $f'$. However, notice
  that whenever a segment of $f$ is split the endpoints of each
  subsegment consist of vertices of $\partial(\F)$ or
  $\partial(\F^*)$. Moreover, exactly two segments in $\partial(\dD)$
  share an endpoint. Thus, the complexity of $\dD$ is $O(m+n)$.
\end{proof}
\subsection{Generating \adg s}\label{implementation:adjacency_graph}
We consider a specific component $F$ of $\F$ and construct its
\adg~$G$. Denote
$F^*\defeq F \setminus \bigcup_{x \in (S \cup T)\cap F} D^*(x)$. Note
that by the analysis in Section~\ref{sec:multi_cc} we can ignore the
influence of $\D_2(x)$ on connected components in $\F$ that do not
contain $x$. We assume that $F^*$ breaks into $k$ maximal connected
components $F^1,\ldots, F^k$. The construction of $G$, along with the
paths in $F$ that correspond to the edges of $G$, is carried out in
two steps. First, for every $F^i$ we generate the portion of $G$,
denoted by $G^i$, whose vertices represent start and target positions
that touch the boundary of $F^i$. Then, we connect between the various
parts of $G$.

We consider a specific connected component $F^i$ of $F^*$ and describe
how the respective portion of the \adg, namely $G^i$, is generated.
We split the start and target positions that share a boundary with
$F^i$ into two subsets:~$B^i$ are those positions for which the
collision disc intersects the boundary of $F$ and $H^i$ are those
positions for which the collision disc floats inside $F$. See
Figure~\ref{fig:graph_generation}. We first handle positions
in~$B^i$. Consider the outer boundary $\Gamma^i$ of
$F^i\setminus \bigcup_{x\in S\cup T}D^*(x)$. We argue that each
$x\in B^i$ can contribute exactly one piece to $\Gamma^i$.

\begin{lemma}\label{lem:gamma_single}
  If $x\in B^i$ then $\partial(\D_2(x))\cap \partial(F^i)$ consists of
  a single component.
\end{lemma}
\begin{proof}
  By contradiction, assume that the intersection consists of two
  maximal connected components. Denote by $y,y'$ two configurations on
  the two components. As $F^i$ consists of a single connected
  component of $F$ there exists a path $\pi_{yy'}\subset F$ from $y$
  to $y'$. Additionally, as $y,y',x$ belong to the same connected
  component of $\F$ there exist two paths---$\pi_{xy}$ from $x$ to $y$
  and $\pi_{xy'}$ from $x$ to $y'$---that lie entirely in
  $D^*(x)$. Thus, the area that is bounded by the three paths
  $\pi_{yy'},\pi_{xy}, \pi_{xy'}$ contains a patch of forbidden space,
  which contradicts the fact the our workspace is a simple polygon.
\end{proof}

For every $x\in B^i$ we arbitrarily select a representative point
$\beta^i(x)\in \partial(\D_2(x))\cap F^i$. We order the points
$\beta^i(x)$ clockwise around $\Gamma^i$, and store them in a circular
list $\L^i$. We now incorporate the remaining start and target
positions $H^i$, namely those positions $x$ for which
$D_2(x)\cap \partial(F)=\emptyset$. Each position in $H^i$ will be
connected either to $\Gamma^i$ or to the boundary of a collision disc
of another position in $H^i$ as follows.  For each $x\in H^i$ we shoot
a vertical ray upwards until it hits $\partial(F^i)$. Denote the point
where the ray hits $\partial(F^i)$ by $c$. If
$c\in \partial(\D_2(x'))$ for some $x'\in H^i, x'\neq x$ then an edge
between $x$ and $x'$ is added to $G^i$.  Otherwise, we let
$\beta^i(x)\defeq c$ and insert it into the circular list $\L^i$
representing the points $\beta^i(x)$ along $\Gamma^i$ collected so
far. After all positions in $H^i$ have been handled in this manner,
for each pair of consecutive points $\beta^i(x'),\beta^i(x'')$ in
$\L^i$ (along $\Gamma^i$) we add an edge in $G^i$ between the vertices
$x'$ and $x''$. (Notice that some of the positions $x$ whose
$\beta^i(x)$ appear in $\L^i$ belong to $H^i$; for example $s_3$ in
Figure~\ref{fig:graph_generation}.)  Finally, the connection between
portions of the \adg~ that represent different parts of $F^*$ is
achieved through positions shared between two sets $B^i,B^j$, for
$i\neq j$.

\subsection{Transforming graph edges into paths in the free space}
There are three different types of transformations depending on how
the edge was created. Let $(x,x')$ be an edge in $G_i$. Consider
Figure~\ref{fig:graph_generation} for an illustration. (i) If both $x$
and $x'$ belong to $H^i$ (see $(s_4,s_3)$ in the figure) then the path
simply consists of the two straight-line segments $\overline{xc}$ and
$\overline{cx'}$. For the remaining two cases we note that if either
vertex, say $x$, is in $B^i$, then part of the path is a simple curve
connecting $x$ to $\beta^i(x)$ within $D^*(x)$ (see the red curves
from $s_1$ and $t_2$ in the figure). We denote this curve by
$\delta_x$. (ii) $x,x'\in B^i$ and the points $\beta^i(x)$ and
$\beta^i(x')$ are consecutive along $\Gamma^i$ (see $(t_1,t_2)$ in the
figure). The path corresponding to the edge $(x,x')$ in this case is a
concatenation of three sub-paths: $\delta_x$, the portion of
$\Gamma^i$ between $\beta^i(x)$ and $\beta^i(x')$ (not passing though
the boundary of any other collision disc), and $\delta_{x'}$. (iii)
$x\in H^i$ and $x'\in B^i$ (see $(s_3,s_2)$ in the figure). The path
is again a concatenation of three paths: the line segment
$\overline{x\beta^i(x)}$, the portion of $\Gamma^i$ between
$\beta^i(x)$ and $\beta^i(x')$ (not passing though the boundary of any
other collision disc), and $\delta_{x'}$.

Notice that for all path types above if a robot $r$ moves from $x$ to
$x'$, $x'$ is not occupied, and all other robots occupy positions only
at $S\cup T\setminus\{x,x'\}$, $r$ will not collide with any other
robot during the motion.

\begin{figure}[t!]
  \vspace{-20pt} \centering \subfloat[ ] {
    \includegraphics[height=0.6\columnwidth] {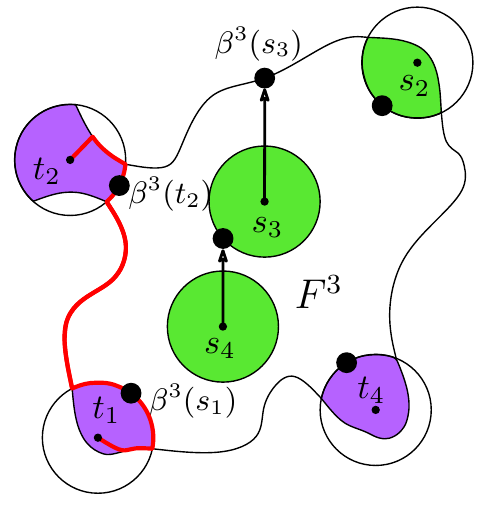}}
  \subfloat[ ]{\raisebox{5mm}{\includegraphics[height=0.3\columnwidth]
      {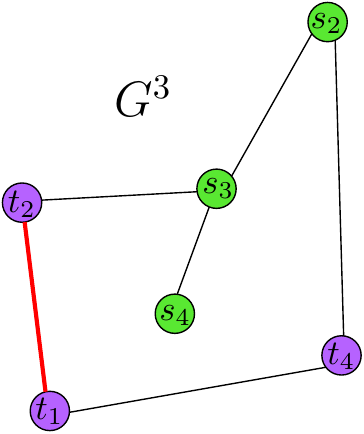}}}
  \caption{(a) An illustration of a component $F^3$ of $F^*$ and the
    structures used for generating the relevant portion of the
    \adg. The boundary positions of $F^3$ consist of
    $B^3:= \{s_2,t_1,t_2,t_4\}$, while the hole positions consist of
    $H^3 := \{s_3,s_4\}$. For every $x\in H^3$ its boundary
    representative $\beta^3(x)$ is illustrated as a large black dot. A
    path between $t_1$ and $t_2$ is illustrated in red. (b) The \adg~
    $G^3$ induced by $F^3$.}
  \label{fig:graph_generation}
\end{figure}

\subsection{Complexity Analysis}
We provide complexity analysis of our algorithm and show that a
solution to the problem can be produced within
$O\left((m+n)\log (m+n)+mn+m^2\right)$ operations, which can be
rewritten as $O\left(n\log n+mn+m^2\right)$.

Recall that the pebble problem solver (Section~\ref{sec:single_cc})
operates in $O(m)$ phases, where in each phase a leaf node is removed
from the spanning tree of $G$. We show, using
Lemma~\ref{lem:complexity_F*} and Lemma~\ref{lem:complexity_D^*}, that
each phase can be transformed into a set of movements for the robots
whose combinatorial complexity is $O(m+n)$. The crucial observation is
that in one phase each edge of the motion graph is used at most
once. Thus the set of robot movements in one phase is bounded by the
complexity of the movements corresponding to all the edges in the
graph together. These comprise $O(m)$ line segments, portions of the
boundaries $\Gamma^i$ (whose complexity is $O(m+n)$ by
Lemma~~\ref{lem:complexity_F*}), and the paths $\delta_x$ inside the
$D^*(x)'s$ (whose complexity is $O(m+n)$ by
Lemma~\ref{lem:complexity_D^*}). A path of the latter type,
$\delta_x$, might be traversed twice: once for reaching $x$ and once
for leaving $x$. However asymptotically all the movements together
have complexity $O(m+n)$.

We note that the cost of generating $\F$, along with its partitions
$\F^*$ and $\dD$, is bounded by $O\left((m+n)\log(m+n)\right)$, due
to~\cite{klps-ujr86}.  We also note that deciding whether a solution
exists for a certain collection of start and target positions can be
carried out in $O((m+n)\log n)$ as follows. We first compute $\F$ in
$O(n\log n)$ time, and within the same time preprocess it for
efficient point location. Then we query the resulting structure with
the $m$ points in $S\cup T$, in $O(\log n)$ time each, and verify that
in every component $F_i$ of $\F$ it holds that $|S_i|=
|T_i|$. Thus, we have the following theorem.

\begin{theorem}
  Let $\W$ be a simple polygon with $n$ vertices and let
  $S=\{s_1,\ldots,s_m\}, T=\{t_1,\ldots,t_m\}$ be two sets of $m$
  points in $\W$. Additionally, assume that for every two distinct
  element $x,x'$ of $S\cup T$ it holds that $\|x-x'\|\geq 4$.  Then,
  given $m$ unlabeled unit disc robots, our algorithm can determine
  whether a path moving the $m$ robots from $S$ to $T$ exists in
  $O\left((m+n)\log n\right)$ time. If a path exists the algorithm
  finds it in $O\left(n\log n+mn+m^2\right)$ time.
\end{theorem}

\section{Separation and solvability}\label{sec:solvability}
Our results from the previous section imply that a separation distance
of~$4$ ensures that the problem always has a solution, assuming that
each connected component contains the same number of start and target
positions. However, for smaller separation values this does not have
to be the case.

We define a magnitude $\lambda$ to be the minimum separation between
start and target positions of unlabeled discs in a simple polygon,
that guarantees that a solution always exists assuming that each
connected component contains the same number of start and target
positions. Our work in the previous sections shows that
$\lambda\leq 4$. The following proposition provides a lower bound for
the value of $\lambda$.
\begin{propo}\label{prop:lambda}
  $\lambda\geq 4\sqrt{2}-2(\approx 3.646)$.
\end{propo}
\begin{proof}
  We describe a concrete example where the value of $\lambda$ has to
  be greater than $4\sqrt{2}-2$ to guarantee solvability. See
  Figure~\ref{fig:small_separation}. The two robots $r_1,r_2$, placed
  at $s_1,s_2$, respectively, need to leave through the corridor of
  width $2$ located below their initial positions in order to reach
  $t_1,t_2$. Denote the separation between the start positions by
  $2+\rho$. We wish to find the maximal value $\rho$ for
  which the problem does not have a solution.

  It is clear that the two robots cannot enter the corridor
  simultaneously. Therefore, if a solution exists there also exists
  one where $r_2$ stays put in $s_2$ until $r_1$ is well down the
  corridor. We describe a path of $r_1$ in which it maintains
  a maximal distance from $s_2$: first $r_1$ has to slide along the
  edge $AB$, when arriving at $B$ it revolves around it and then
  slides down along the other edge connected to $B$. This path has a
  circular arc which is induced by the rotation about $B$ (depicted by
  the red arrow).  The dashed red arc depicts the trace of the
  boundary of $r_1$ throughout this motion.

  Denote by $M$ the tangent point between $r_2$ at $s_2$ and the edge
  $CD$. Note that the robots would collide when $r_1$ is rotated
  around $B$ in case that $\|B-s_2\|<3$. Thus, it must be that
  $\|B-M\|\geq 2\sqrt{2}$. As the segment $AD$ consists of the
  subsegments $AB,BM,MD$, it follows that
  \begin{dmath*}[compact]
    4+\rho  = {\|A-D\|} \\ = {\|A-B\|+\|B-M\|+\|M-D\|} \\ \geq
    {1+\frac{\rho}{2}+2\sqrt{2}+1} \\ = {2+2\sqrt{2}+\frac{\rho}{2}},
  \end{dmath*}
  \noindent and we conclude that $\rho\geq 4\sqrt{2}-4$, and
  also $\lambda \geq \rho + 2 =4\sqrt{2}-2$.
\end{proof}

\begin{figure}[h]
  \centering
  \includegraphics[width=0.8\columnwidth]
  {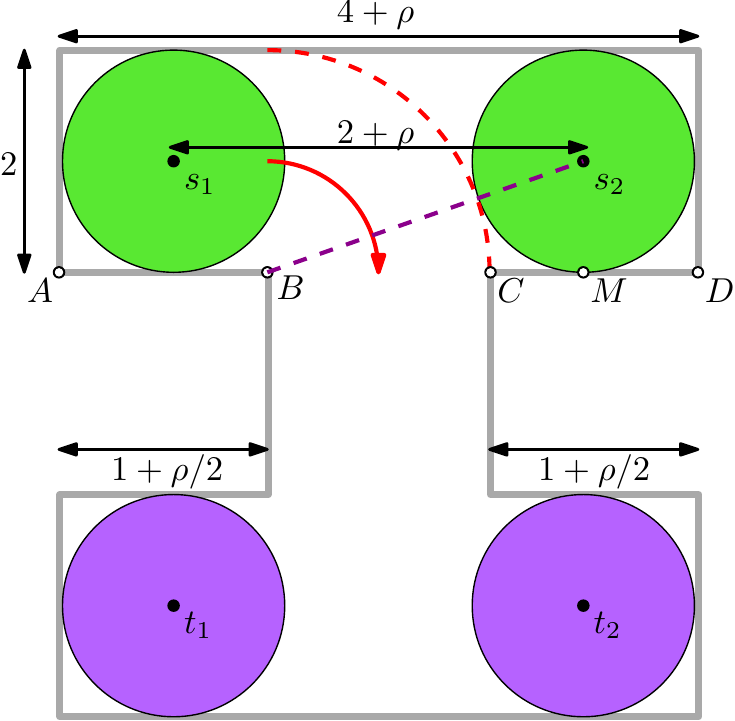}
  \caption{It follows from our paper that when the start and goal
    positions are well separated, then there is always a solution when
    each free-space component has the same number of start and goal
    positions. However, it is not true when the separation condition
    is not met. In the given example, the two robots need to leave the
    through a corridor of width $2$ located below their initial
    positions. This is only possible when $\rho\geq 4\sqrt{2}-2$, as
    otherwise the robots will not be able to get to the corridor
    without colliding with each other. See proof of
    Proposition~\ref{prop:lambda} for more
    details.}\label{fig:small_separation}
\end{figure}

\section{Open problems and future work}\label{sec:discussion} 
We have studied a basic variant of the multi-robot motion-planning
problem, where the goal is to find collision-free motions that bring a
given set of indistinguishable unit discs in a simple polygon to a
given set of target positions. Under the condition that the start and
target positions are separated from each other by a distance of at
least $4$, we developed an algorithm that solves the problem in time
polynomial in the complexity of the polygon as well as in the number
of discs: quadratic in the number of robots and near-linear in the
complexity of the polygon.

Our result should be contrasted with the labeled counterpart of the
problem, which is \np-hard~\cite{sy-snp84}. In the \np-hardness proof
the discs have different radii, however, and there is no restriction
on the separation of the start and target position. Thus one of the
main open questions resulting from our study is to settle the
complexity of the unlabeled problem without this extra separation
condition. Very recently, we have made progress in this direction, by
showing that the general unlabeled problem is
\pspace-hard~\cite{sh-pspace14}. It should be noted that this proof
applies to a slightly different setting that consists of unit-square
robots translating amidst polygonal obstacles.

A natural question that arises is what happens when the separation of
start or target positions is equal to $2+\rho$, where
$0<\rho < 2$? Would it be possible to design an algorithm,
whose running time is polynomial in $m$ and $n$, and also depends in
some manner on $\rho$? We believe that the work by Alt et
al.~\cite{alt-approx92} would be a good starting point for this line
of work. Following our discussion in Section~\ref{sec:solvability}, it
will also be interesting to find the exact threshold above which a
problem is always solvable.
 
\bibliographystyle{IEEEtran}
\bibliography{bibliography}

\end{document}